\documentclass[conference]{IEEEtran}

\usepackage{amsmath,amssymb}
\usepackage{subfigure}
\usepackage{graphicx,graphics,color,psfrag}
\usepackage{cite,balance}
\usepackage{caption}
\allowdisplaybreaks
\usepackage{accents}
\usepackage{amsthm}
\usepackage{bm}
\usepackage[english]{babel}
\usepackage{multirow}
\usepackage{enumerate}
\usepackage{cases}
\usepackage{stfloats}
\usepackage{dsfont}
\usepackage{color,soul}
\usepackage{amsfonts}
\usepackage{cite,graphicx,amsmath,amssymb}
\usepackage{subfigure}
\usepackage{fancyhdr}
\usepackage{hhline}
\usepackage{graphicx}
\usepackage{array,color}
\usepackage{booktabs}
\usepackage{diagbox}
\usepackage{indentfirst}
\usepackage{url}
\usepackage[ruled,vlined]{algorithm2e}

\usepackage{amsmath}

\usepackage{lipsum}

\newtheorem{assumption}{Assumption}
\newtheorem{lem}{Lemma}

\theoremstyle{definition}

\newtheorem{rem}{Remark}
\newtheorem{theo}{Theorem}
\usepackage{lipsum}

\begin{document}
	\title{Low-Latency Cooperative Spectrum Sensing via Truncated Vertical Federated Learning}
	
\IEEEoverridecommandlockouts{
	\author{\IEEEauthorblockN {Zezhong Zhang\IEEEauthorrefmark{1}\IEEEauthorrefmark{2},
			Guangxu Zhu\IEEEauthorrefmark{3},  Shuguang Cui\IEEEauthorrefmark{1}\IEEEauthorrefmark{2}\IEEEauthorrefmark{3}}
		
		\IEEEauthorblockA{
			\IEEEauthorrefmark{1}The Future Network of Intelligence Institute (FNii), Shenzhen, China\\
			\IEEEauthorrefmark{2}The Chinese University of Hong Kong (Shenzhen), Shenzhen, China\\
			\IEEEauthorrefmark{3}Shenzhen Research Institute of Big Data, Shenzhen, China\\
			Contact Email:  zhangzezhong@cuhk.edu.cn}	
}}
	\maketitle

\begin{abstract}

In recent years, the exponential increase in the demand of wireless data transmission rises the urgency for accurate spectrum sensing approaches to improve spectrum efficiency. The unreliability of conventional spectrum sensing methods by using measurements from a single \emph{secondary user} (SU) has motivated researches on \emph{cooperative spectrum sensing} (CSS). In this work, we propose a \emph{vertical federated learning} (VFL) framework to exploit the distributed features across multiple SUs without compromising the data privacy. However, the repetitive training process in VFL faces the issue of high communication latency. To accelerate the training process, we propose a \emph{truncated vertical federated learning} (T-VFL) algorithm, where the training latency is highly reduced by integrating the standard VFL algorithm with a channel-aware user scheduling policy. The convergence performance of T-VFL is provided via mathematical analysis and justified by simulation results. Moreover, to guarantee the convergence performance of the T-VFL algorithm, we conclude three design rules on the neural architectures used under the VFL framework, whose effectiveness is proved through simulations.

\end{abstract}

\section{Introduction}
Since the emergence of wireless communication technology, the reuse of radio spectrum resources has remained as an inherent issue. The usage of the spectrum resources is regulated by the \emph{Federal Communications Commission} (FCC), according to whom cognitive radio is defined as an intelligent device that can sense the environment and autonomously adjust its radio operating parameters. A key functionality provided by cognitive radio is \emph{opportunistic spectrum access} (OSA), where non-licensed \emph{secondary users} (SUs) are allowed to utilize the idle licensed spectrum resources to maximize its performance objectives. To this end, accurate and informative spectrum sensing results are crucial. However, conventional spectrum sensing by cognitive radio devices faces the issue of reliability of measurements made by a single user.
The signals emitted from the \emph{primary users} (PUs) are coupled with multi-path fading and shadowing, making the SUs unable to detect the PUs' states from its measurements in a certain moment or location \cite{Hussain2009}.  

To provide more reliable detection on the PUs' states, \emph{cooperative spectrum sensing} (CSS) has been proposed in the existing literature, divided into two categories according to the decision modes, i.e., hard decision and soft decision. In the hard decision mode, each SU makes binary decision based on threshold levels and then sends it to the server for final decision. On the other hand, intermediate features are extracted and uploaded to the server for the soft decision mode, where every local measurement contributes to the final decision. The informative intermediate features make the soft decision mode more accurate and robust compared to the hard decision mode. A straightforward approach under the soft decision mode is to upload the exact \emph{received signal strength} (RSS) to the server for a centralized decision. An optimal linear fusion algorithm was also proposed in \cite{Quan2008}. However, uploading the raw measurements encounters privacy issues, e.g., the server can infer the SUs' locations based on the local RSS information if uploaded. Addressing this issue has motivated researchers to apply federated learning that preserves data privacy to CSS. As originally proposed for distributed learning, the federated learning framework offloads partial data processing procedures to the SUs \cite{Mcmahan17} for privacy protection, which also helps exploit the distributed computation resources at the SUs. In this work, we adopt the \emph{vertical federated learning} (VFL) framework \cite{YQ2019} in view of the fact that at each time slot, the RSSs measured at each SU can be considered as partial features of the PUs' states. In the existing literature, early attempts to apply VFL to the CSS task have been made in \cite{Gao2019,Zhang2020}. The authors in \cite{Gao2019} adopt the hard decision mode with scalar local output. In \cite{Zhang2020}, a VFL design with the soft decision mode is considered and a RNN-based model is used to exploit the consistency of data samples in the time domain to improve the detection accuracy. The convergence performance of the standard VFL algorithm is analyzed in \cite{YQarxiv}.
However, in the above works \cite{Gao2019,Zhang2020,YQarxiv}, the authors consider the case that all SUs participate in each round, where all local models and the central model are trained with the aid of data exchange between the server and the SUs. In this case, high communication latency takes place once if some wireless links experience deep channel fadings. The repetitive communication during the training process makes communication latency the bottleneck \cite{Mcmahan17}. Therefore, the main challenge of deploying the VFL framework in a wireless network is how to achieve fast convergence under limited communication and computation resources. 

In view of the drawbacks of the existing works, in this paper we propose a VFL framework to detect the PUs' states, including their transmit power levels and locations.
To minimize the total training latency, which consists of the communication and computation latency through the whole training process, we propose a \emph{truncated VFL} (T-VFL) algorithm by integrating the standard VFL algorithm with a channel-aware user scheduling policy. The convergence performance of the T-VFL algorithm is provided through both analysis and simulations. Moreover, to cope with the scheduling policy, we also provide some design rules on the neural architectures used under the VFL framework to ensure fast convergence. The key contributions of this work is concluded as follows.

\begin{itemize}
	\item {\bf Low-Latency VFL via User Scheduling:} The standard VFL algorithm	requires updating the whole \emph{deep neural network} (DNN) in each round, which brings in high communication latency when some SUs experience deep channel fadings. To tackle this problem, we propose the T-VFL algorithm. By scheduling SUs with channel gains above a given threshold, the per-round communication latency can be highly reduced. On the other hand, the user scheduling policy has a side effect on slowing down the convergence speed that more training rounds are required to achieve a given accuracy. The effect of the proposed user scheduling policy on the convergence speed is also described mathematically in virtual of learning theory.
	
	\item {\bf Neural Architecture Designs:} Based on the convergence analysis, we draw an upper bound on the total training latency of the T-VFL algorithm. The bound clearly shows the effect of different wireless parameters, while it is only valid when all assumptions hold. Therefore, instead of minimizing the analytical bound, we give some design rules on the neural architectures to enhance the robustness of its convergence performance. We show through simulations that by using the proposed neural architectures, the T-VFL algorithm outperforms the standard VFL algorithm with an arbitrary activation ratio. 
\end{itemize}

	\section{Collaborative Spectrum Sensing Model}\label{SystemModel}
	In this section, we present the model of a distributed wireless communication system, design the CSS algorithm, and describe its implementation in the system.

	\subsection{Distributed Wireless Communication System}\label{communication}
	We consider a cellular network with an arbitrary number of PUs and $K$ SUs. In each time slot, the licensed spectrum is partially or fully occupied by the PUs for communication. To fully exploit the spectrum resources without interfering the communication quality of the PUs, the SUs can reuse the spectrum once it is idle or the active PUs are far away. To this end, the SUs need to first detect the states of the PUs, i.e., their locations and transmit power levels.

	Following the idea of cognitive radio, the states of the PUs can be inferred by exploiting the RSSs at the SUs \cite{Hussain2009,Quan2008}, which are interpreted as implicit features. For the RSSs at each SU, the signal propagation paths are determined by the environment between the PUs and the SU. Hence the RSSs follow distinct distributions at different SUs, indicating that the features observed at different SUs are not fully overlapped with each other. Intuitively, by intelligently combining the features across all SUs, the PUs' states can be detected with higher accuracy than using features from a single SU. In practice, the correlation among features at different SUs is usually unknown, which calls for a general solution to the CSS problem by taking all the features equally as direct input information for decision. By exploiting the advantage of supervised learning techniques which train neural networks using labeled data without any constraint or assumption on the input features. we present a VFL framework and its deployment in the wireless system in the sequel.

	\subsection{Collaborative Spectrum Sensing Algorithm}
	\subsubsection{Machine Learning Model} 
	Suppose the data collection process lasts for $M$ time slots where each SU records the RSSs and its own location at the current time slot as the local features. In this case, the local dataset at each SU $k$ is denoted as $\{\mathbf{x}_i^k \in \mathbb{R}^{d_k\times1}\}_{i=1}^M$, consisting of $M$ feature vectors of dimension $d_k$. The PUs' states are identical to all SUs at a given time slot, while the propagation paths depend on the environment around each SU, making the local features distinct across the SUs.  To improve the accuracy of estimating the state of the PUs, we propose the following VFL framework, where the features from all SUs are concatenated as a complete data sample for inference.
	
	Denote the \emph{complete feature vector} sampled at time slot $i$ as $\mathbf{x}_i = [{\mathbf{x}_i^1}^T, \dots, {\mathbf{x}_i^K}^T]^T$, and suppose the associated label $y_i$ is available at the server. Hence the \emph{complete data sample} $i$ is given as $\mathcal{D}_i \triangleq \{\mathbf{x}_i, y_i \}$, and the global dataset is denoted as $\mathcal{D} = \{\mathcal{D}_i\}_{i=1}^M$. The idea of VFL is to train a DNN which is divided into $K+1$ parts and distributed across the SUs and the server. Each SU owns a local model with parameters denoted as $\boldsymbol{\theta}_k$ and the server owns a central model with parameters denoted as $\boldsymbol{\theta}_0$. All the $K+1$ models constitute the whole DNN, given as $\boldsymbol{\Theta} = [\boldsymbol{\theta}_0, \dots, \boldsymbol{\theta}_K]$. The cooperative training problem can be formulated as 
	\begin{align}
	(\mathbf{P1}) \ \min_{\boldsymbol{\Theta}} \mathcal{L}(\boldsymbol{\Theta}; \mathcal{D}) \!\triangleq\! \frac{1}{M} \!\sum\limits_{i=1}^M\! f(\boldsymbol{\theta}_0, \dots, \boldsymbol{\theta}_K; \mathcal{D}_i) \!+\! \lambda \!\sum\limits_{k=0}^K\! \gamma(\boldsymbol{\theta}_k),\nonumber
	\end{align}
	where $\gamma(\cdot)$ denotes the regularization function and $\lambda$ is the hyperparameter.
    
    Training the whole DNN requires multiple rounds of communication between the server and SUs, where each round is divided into a forward and backward transmission phases. We give an example to illustrate the training process in each round. Without loss of generality, we use \emph{mean squared error} (MSE) as the loss function and the round index $t$ is omitted.
    \begin{itemize}
    	\item {\bf Forward Phase:} For each data sample $i$, each SU feeds its local feature vector $\mathbf{x}_i^k$ into the local model $\boldsymbol{\theta}_k$, which outputs an intermediate result $\mathbf{p}_i^k = \mathcal{C}_1(\boldsymbol{\theta}_k; \mathbf{x}_i)$. Then the intermediate results at all SUs are uploaded to the server as the input to the central model $\boldsymbol{\theta}_0$. The output of the central model is denoted as $\hat{y}_i = \mathcal{C}_2(\boldsymbol{\theta}_0; \mathbf{p}_i^1,\dots, \mathbf{p}_i^K)$, based on which the function loss is expressed as 
    	\begin{align}\label{functionLoss}
    	\mathcal{L}(\boldsymbol{\Theta}; \mathcal{D}) = -\frac{1}{M} \sum\limits_{i=1}^M \|y_i - \hat{y}_i\|^2 + \lambda \sum\limits_{k=0}^K \gamma(\boldsymbol{\theta}_k).
    	\end{align}  
    	\item {\bf Backward Phase:} With the loss value calculated according to \eqref{functionLoss}, the gradient of the central model can be obtained at the server as $\mathbf{g}(\boldsymbol{\theta}_0) = \frac{\partial \mathcal{L}(\boldsymbol{\Theta}; \mathcal{D})}{\partial \boldsymbol{\theta}_0}$. To calculate the local gradient for the model at SU $k$, whose expression is given as $\mathbf{g}(\boldsymbol{\theta}_k) = \frac{\partial \mathcal{L}(\boldsymbol{\Theta}; \mathcal{D})}{\partial \mathbf{p}_i^k} \cdot \frac{\partial \mathbf{p}_i^k}{\partial \boldsymbol{\theta}_k}$, the server sends the information of $\frac{\partial \mathcal{L}(\boldsymbol{\Theta}; \mathcal{D})}{\partial \mathbf{p}_i^k}$ to SU $k$ and $\frac{\partial \mathbf{p}_i^k}{\partial \boldsymbol{\theta}_k}$ can be computed locally. Then all the SUs and the server update their model using gradient descent as
    	\begin{align}
    	\boldsymbol{\theta}_k \leftarrow \boldsymbol{\theta}_k - \eta \mathbf{g}(\boldsymbol{\theta}_k), \quad k = 0,\dots, K,
    	\end{align}
    	with step-size $\eta$.
    \end{itemize}

	\subsection{Transmission Model}\label{transmissionModel}
	As described in the previous subsection, the whole DNN is updated in each round, which requires frequent information exchange via wireless links. Specifically, all devices upload the intermediate results calculated using the latest local models to the server in the uplink, and the server sends the gradient of the input layer in the central model to SUs in the downlink to facilitate the local gradient calculation and model update. We consider the downlink transmission the high SNR case with negligible error and only focus on the uplink phase.
	
	Consider the uplink phase of an arbitrary round. Assume that the local outputs of each SU are of the same dimension $d$, and each output symbol is represented by $q$ bits after quantization. By using BPSK modulation, $D\!=\!qMd$ symbols need to be transmitted at each SU when $M$ data samples are used for training. The channel coefficient from server to SU $k$ in round $t$ is denoted as $H_{k,t} \!=\! \sqrt{\rho_{k,t}}h_{k,t}$, where $\rho_{k,t} \!=\! \phi_{k,t}d_{k,t}^{-\kappa}$ is the large-scale fading coefficient with shadowing factor $ \phi_{k,t}$ and pathloss exponent $\kappa$, and $h_{k,t} \in \mathcal{CN}(0, 1)$ represents the small-scale fading coefficient. Then the $i$-th symbol received at the server in round $t$, denoted as $\mathbf{y}_{k,t}^{(i)} \in \mathbb{C}^{D\times 1}$, is given as
	\begin{align}
	\mathbf{y}_{k,t}^{(i)} = H_{k,t}p_{k,t}\mathbf{s}_{k,t} + \mathbf{z}_{k,t}, \quad 1 \le k \le K, \ t \ge 1,
	\end{align}
	where $\mathbf{s}_{k,t}\in\mathbb{C}^{D\times 1}$ is the symbols transmitted by device $k$, in which each element has unit variance. The complex Gaussian vector $\mathbf{z}_{k,t}\in\mathbb{C}^{D\times 1}$ is the noise vector, with each element of zero mean and variance $\sigma^2$. Let $B$ be the bandwidth allocated to each SU for transmission, the achievable data transmission rate of the link from SU $k$ to the server is given by 
	\begin{align}\label{wirelessModel}
	r_{k,t} = B \log_2 \bigg(1+\frac{p_{k,t}^2 |H_{k,t}|^2}{\sigma^2} \bigg), 
	\end{align}
	where $p_{k,t}$ is the power coefficient. The transmission of each SU is subject to an average power constraint:
	\begin{align}\label{powerConstraint}
	\mathbb{E}_t[p_{k,t}^2] \le P,
	\end{align}
	for a given constant $P$, which can be further simplified as $\mathbb{E}_t[p_{k,t}^2] = P$.
	
	\subsection{User Scheduling Policy} 
	Note that the server needs to collect the intermediate results from all SUs before subsequent computation using the central model. Therefore, the communication latency is determined by the SU with the lowest transmission rate. Under this circumstances, it is power-efficient to align the transmission rate of all SUs since scheduling other SUs with high transmission rates does not accelerates process. To this end, we propose the following user scheduling policy evolved from the so-called \emph{truncated channel inversion} scheme \cite{Broadband} to align the transmission rates and avoid high communication latency due to severe channel fadings. Specifically, the power coefficient in \eqref{wirelessModel} is given as
	\begin{align}\label{truncation}
	p_{k,t} = \left\{ \begin{aligned}
	&\frac{{\sqrt {P_{k,t}^{\text{rx}}} }}{{\sqrt{\rho_{k,t}}h_{k,t}}},& {\left| {h_{k,t}} \right|^2} \geq {G_{k,t}},  \\
	&0,& \ \ {\left| {h_{k,t}} \right|^2} < {G_{k,t}}, \\ 
	\end{aligned}  \right.
	\end{align}
	where $G_{k,t}$ is called the truncation threshold. To fulfill the average power constraint in \eqref{powerConstraint}, it can be obtained that
	\begin{align}
	{P_{k,t}^{\text{rx}}} = \frac{\rho_{k,t} P}{\mathsf{Ei}(G_{k,t})},
	\end{align}
	where $\mathsf{Ei}(X) = \int_{X}^{\infty}\frac{1}{x}\exp(-x)dx$, and $\epsilon_{k,t} = \mathsf{Ei}(G_{k,t})$ represents the activation ratio of SU $k$. Note that ${P_{k,t}^{\text{rx}}}$ determines the received SNR and the transmission rate, we require 
	\begin{align}\label{Alignment}
	\frac{\rho_{1,t}}{\mathsf{Ei}(G_{1,t})} = \frac{\rho_{k,t}}{\mathsf{Ei}(G_{k,t})}
	\end{align} 
	Without loss of generality, we suppose SU $1$ suffers the deepest channel fading. Moreover, the large-scale fading coefficients $\{\rho_{k,t},\ \forall k,t\}$, which keep unchanged over multiple time slots, are usually easy to be obtained at the server a prior. Then we control the activation ratio of the SUs by tuning the truncation threshold $G_{1,t}$ of SU $1$. The truncation thresholds of other SUs can be adjusted according to \eqref{Alignment} by numerical search given that $\mathsf{Ei}(\cdot)$ is a monotonic function.
	
	The SUs with channel gain lower than their truncation thresholds are silenced in the current round. To facilitate subsequent calculation, the server uses the latest intermediate results received from the silenced SUs. According to the above elaboration, the proposed user scheduling policy prevents high communication latency in each round. However, by cutting down the number of active SUs, it also has an effect of slowing down the convergence of the training process, which is mathematically depicted in Theorem \ref{ConvergenceTheo} in the sequel.

	\section{Convergence Analysis}\label{ConvAnalysis}
	To give a tractable analysis on the VFL framework, we first adopt the following two assumptions.
	\begin{assumption}[Lipschitz Smoothness]
		\begin{align}\label{smoothness}
		\nabla\! \mathcal{L}(\boldsymbol{\Theta}_{t+1}) \!-\! \nabla\! \mathcal{L}(\boldsymbol{\Theta}_{t}) \!\le\! L \| \boldsymbol{\Theta}_{t+1} \!- \boldsymbol{\Theta}_{t} \|, \quad  \forall \boldsymbol{\Theta}_{t}, \boldsymbol{\Theta}_{t+1}.
		\end{align}
	\end{assumption}
	\begin{assumption}[Convexity]
		\begin{align}\label{convexity}
		\nabla\! \mathcal{L}(\boldsymbol{\Theta}_{t+1}) \!-\! \nabla\! \mathcal{L}(\boldsymbol{\Theta}_{t}) \!\ge\! \mu \| \boldsymbol{\Theta}_{t+1} \!- \boldsymbol{\Theta}_{t} \|, \quad \forall \boldsymbol{\Theta}_{t}, \boldsymbol{\Theta}_{t+1}.
		\end{align}
	\end{assumption}

	\begin{assumption}[Limited Range]
	\begin{align}\label{limitedRange}
	|\mathcal{L}(\boldsymbol{\Theta}_{t+1}) - \mathcal{L}(\boldsymbol{\Theta}_{t})| \le C , \quad \forall \boldsymbol{\Theta}_{t}, \boldsymbol{\Theta}_{t+1}.
	\end{align}
    \end{assumption}

    \begin{assumption}[Limited Variance]
    By using constant batch size, the gradient variance is bounded by
    	\begin{align}\label{limitedVariance}
    	\text{Var}[\nabla \mathcal{L}(\boldsymbol{\Theta}_{t})] \le c, \quad \forall \boldsymbol{\Theta}_{t}.
    	\end{align}
    \end{assumption}
	
	By exploiting the above assumptions, we have the following theorem on the convergence rate of the VFL algorithm.
	\begin{theo}[Convergence Rate]\label{ConvergenceTheo}
		By scheduling an active set of users $\mathcal{K}_A(t)$ to participate in the learning process in each round, the convergence rate of the VFL algorithm is given as
		\begin{align}\label{ConvRate}
		&\mathbb{E}[\mathcal{L}(\boldsymbol{\Theta}_{t+1}) - \mathcal{L}(\boldsymbol{\Theta}^\ast)] \nonumber\\
		\le &\left[ 1 - \frac{\mu v}{L} \right]^{t+1}  \mathbb{E}\left[ \mathcal{L}(\boldsymbol{\Theta}_0) - \mathcal{L}(\boldsymbol{\Theta}^\ast) \right] + \frac{c}{2\mu v}.
		\end{align} 
		with step-size $\eta \le \frac{1}{L}$. The factor $v$ is a lower bound given as
		\begin{align}\label{decompose}
		v \le v_{0,t} + \sum\limits_{k\in\mathcal{K}_A} v_{k,t}, \quad \forall t,
		\end{align}
		where $v_{k,t} = \frac{\| \nabla_k \mathcal{L} (\boldsymbol{\Theta}_n) \|^2}{\sum\limits_{k\in \mathcal{K}} \|\nabla_k \mathcal{L}(\boldsymbol{\Theta}_n)\|^2}$.
	\end{theo}
    \begin{proof}
    	Please refer to Appendix A.
    \end{proof}
	
	Theorem 1 shows the convergence properties of the training process. The coefficient $1 - \frac{\mu v}{L}$ describes the speed of the reduction on $\mathbb{E}[\mathcal{L}(\boldsymbol{\Theta}_{t}) - \mathcal{L}(\boldsymbol{\Theta}^\ast)]$, the gap between the current expected loss and the optimal expected loss.  The bound $v$ explicitly depicts the effect of user scheduling on the convergence rate, and thereby called as the \emph{effective activation level}. Note that $\mu \le L$, and $0 \le v \le 1$, which makes $0 \le 1-\frac{\mu v}{L} \le 1$. A larger $v$, which corresponds to a higher activation ratio of the SUs, reduces the coefficient $1 - \frac{\mu v}{L}$ and pushes it from $1$ to $1-\frac{\mu}{L}$, resulting in an increase on the convergence rate of the VFL algorithm.

	\section{Neural Architecture Designs for Fast Convergence}\label{architectureDesign}
	In this section, we first derive the per-round latency and the number of rounds required to achieve a given accuracy, based on which the expression of the total training time is also obtained. Then we propose neural architecture designs, which are able to accelerate the VFL algorithm when combined with the truncated channel inversion scheme. The effectiveness of the designs are verified through simulations in the next section.
	  
	{\bf Per-round Latency:} By adopting the truncated channel inversion scheme in Sec. \ref{transmissionModel}, the per-round communication latency during the VFL training process is lower bounded by
	\begin{align}
	T_\mathsf{comm} \le \frac{qMd}{B \log_2(1+\frac{P}{\sigma^2}\frac{\rho_{1,t}}{{\mathsf{Ei}(G_{1,t})}})},
	\end{align}
	which is the communication latency of transmission from SU $1$ to the server. The computation time at the SUs, denoted by $T_\mathsf{comp}$, is proportional to the size of local model and thereby considered to be constant in each round. The total latency of each communication round is then given as 
	\begin{align}
	T_\mathsf{total} = T_\mathsf{comm}+T_\mathsf{comp}.
	\end{align}

	{\bf Number of Rounds:} 
	According to Theorem \ref{ConvergenceTheo}, we have the following lemma on the number of rounds required to achieve a given accuracy.
	\begin{lem}
		To achieve $\varepsilon$-accuracy, defined as $\mathbb{E}[\mathcal{L}(\boldsymbol{\Theta}_{t}) - \mathcal{L}(\boldsymbol{\Theta}^\ast)] \le \varepsilon$, the expected number of rounds required for the training process is upper-bounded by
		\begin{align}
		N_{\mathsf{expect}} \le \frac{\ln \left[\frac{\varepsilon}{C} - \frac{c}{2\mu vC} \right] }{\ln \left[1-\frac{\mu v}{L}  \right] }
		\end{align}		
	\end{lem}
    \begin{proof}
    	The result can be directly derived from \eqref{ConvRate}.
    \end{proof}
    
    According to the above derivation, the total training time can be upper bounded by
    \begin{align}\label{TotalBound}
    T_{\mathsf{expect}} \le& N_{\mathsf{expect}} T_
    \mathsf{round} \nonumber\\
    =& \!\left[\!\frac{qMd}{B\! \log_2\!\left(\!1\!+\!\frac{P}{\sigma^2}\frac{\rho_{1,t}}{{\mathsf{Ei}(G_{1,t})}}\!\right) }\!+\! T_\mathsf{comp} \!\right]\!\! \times\! \frac{\ln \!\left[\!\frac{\varepsilon}{C} \!-\! \frac{c}{2\mu vC} \!\right] }{\ln \!\left[1\!-\!\frac{\mu v}{L}  \right] }.
    \end{align}
    This upper bound clearly describes the effects of wireless parameters, e.g., the truncation threshold $G_t$, the effective activation level $v$ and the accuracy $\varepsilon$, on the total training time of the VFL algorithm. Moreover, it provides a feasible approach to reduce the training time by minimizing the upper bound in \eqref{TotalBound}. However, direct optimization on the theoretical bound is not an appropriate solution due to the following several reasons. Firstly, the effects of some parameters on the total training time are complicated. For example, it can be observed that the truncation threshold $G_t$ plays a key role. It not only scales the SNR with a factor $\frac{1}{{\mathsf{Ei}(G_{1,t})}}$, but also has an effect on the effective activation level $v$ which is difficult to characterize, making the optimization problem intractable. Secondly, the theoretical bound is valid only when the assumptions in Sec. \ref{ConvAnalysis} strictly hold, which may not be exactly the case in practice. Moreover, it is difficult to obtain the knowledge of hyperparameters $\mu$ and $L$ a priori. In view of the above facts, we propose a neural architecture design to accelerate the training process for general cases as follows.

    {\bf Design Rule 1:} The dimension of the local output features $d$ determines the communication overhead in each round, which thereby should be chosen as small as possible. To be specific, $d$ can be set slightly larger than the number of features necessary for subsequent decision at the server, which can be estimated empirically according to the type of the task and data samples. 
    
    {\bf Design Rule 2:} The weight of $v_{0,t}$ in the effective activation level $v$ should be set as the largest value in its feasible range. Since $v_{0,t} = \frac{\| \nabla_0 \mathcal{L} (\boldsymbol{\Theta}_n) \|^2}{\sum\limits_{k\in \mathcal{K}} \|\nabla_k \mathcal{L}(\boldsymbol{\Theta}_n)\|^2}$ depends on the gradient power of the central model at the server, which cannot be controlled during the training process, we increase the weight of $v_{0,t}$ in an intuitive way by using neural networks with larger size at the server. The feasible range of the central model size is further discussed in Remark \ref{SizeRange}.
    
    {\bf Design Rule 3:} Given that $v_{0,t}$ is dominant in $v$, an increase on the truncation threshold $G_t$, or a reduction on the activation ratio, has a positive effect on shortening the total training latency unless it comes to the extreme case where the probability that all SUs are truncated tends to $1$.

    The principles behind the design rules 2 and 3 are further elaborated as follows.  To accelerate the training process by the means of user scheduling, a higher truncation threshold $G_t$ leads to lower communication latency and a smaller effective activation level $v$ which requires more training rounds to achieve the $\varepsilon$-accuracy. This makes the choice on $G_{1,t}$ very crucial to controlling the training time. However, as we can observe from \eqref{decompose}, the presense of $v_{0,t}$ is independent of user scheduling, which corresponds to the property of the VFL framework that the central model is always updated in each round. Therefore, a larger $v_{0,t}$ can alleviates the effect of user scheduling on the effective activation level. Specifically, when the weight of $v_{0,t}$ is dominant in the RHS of \eqref{decompose}, it only requires a few additional rounds to achieve the $\varepsilon$ accuracy if we increase the truncation threshold $G_{1,t}$. In this sense, it is favorable to keep $v_{0,t}$ dominant during the training process so that we can tune the truncation threshold $G_{1,t}$ within a large range to control the communication latency without affecting the number of rounds much.
    

    \begin{rem}[Feasible Range of the Central Model Size]\label{SizeRange}
    	When increasing the size of central model at the server according to the Design Point 2, it usually requires smaller local models to prevent over-fitting \cite{Overfitting}. However, a too small local model can lose too much information of the input feature and distort the final decision at the server. Therefore, we should only manipulate the size of the central model within a feasible range to achieve fast convergence without distorting the learning performance. When there are some empirical neural architectures with similar learning capacity, the Design Rule 2 indicates that choosing the one with the largest central model size offers the least number of training rounds to converge under the same user scheduling policy.
    \end{rem}

	\section{Simulations}\label{simulation}
	To check the performance of the VFL algorithms on spectrum sensing tasks and validate the effectiveness of the proposed neural architecture design, we consider a typical wireless network with $400$m $\times$ $400$m area, where $4$ SUs keep detecting the states of $2$ PUs, as illustrated in Fig. \ref{simScenario}. Each PU has $3$ optional transmit power levels $\{1, 2, 3\}$. The pathloss exponent is $4$ and a $3$dB shadowing effect is considered. Without loss of generality, we assume SU 1 experiences the deepest large-scale fading and the lowest activation ratio. The SUs collect the RSS and their own locations as local feature vectors across multiple time slots, where the channel states and the locations of all SUs and PUs keep to be constant within $1$ time slot. Each time slot is further divided into $200$ mini-slots, where the RSS in each mini-slot is recorded as a feature at each SU. For each SU, a local feature vector is of $203$ dimensions which consists of the SU's own 3D location and $200$ RSSs within $1$ time slot. The label of each data sample has $8$ dimensions, consisting of the transmit power levels and the 3D locations of the 2 PUs. The distribution of the users in the $X-Y$ plane are illustrated in Fig. \ref{simScenario}, where each user keeps moving within a rectangular region. The height of each location is a function of its coordinate in the $X-Y$ plane as $z = 10\cdot[\sin(\frac{x}{100}) + \cos(\frac{y}{100})]$. We generate $60000$ data samples, divided into a training set with $50000$ samples and a testing set with $10000$ samples. \emph{Fully-connected networks} (FCNs) are used for models at each SU and the server, and the models at all SUs are identical. All tasks run on an Nvidia 1070Ti GPU, and we assume the computation capability of each SU is $\frac{1}{4}$ of the server's computation capability. 
	Since deploying the centralized learning algorithm in a distributed scenario incurs infinite communication latency, we using VFL algorithm with a large activation ratio $\epsilon_{1,t} = 0.9$ as an approximation of the centralized learning algorithm, which is also used as the baseline for comparison. 
	
	\begin{figure}[t]
		\centering
		\includegraphics[height=4cm]{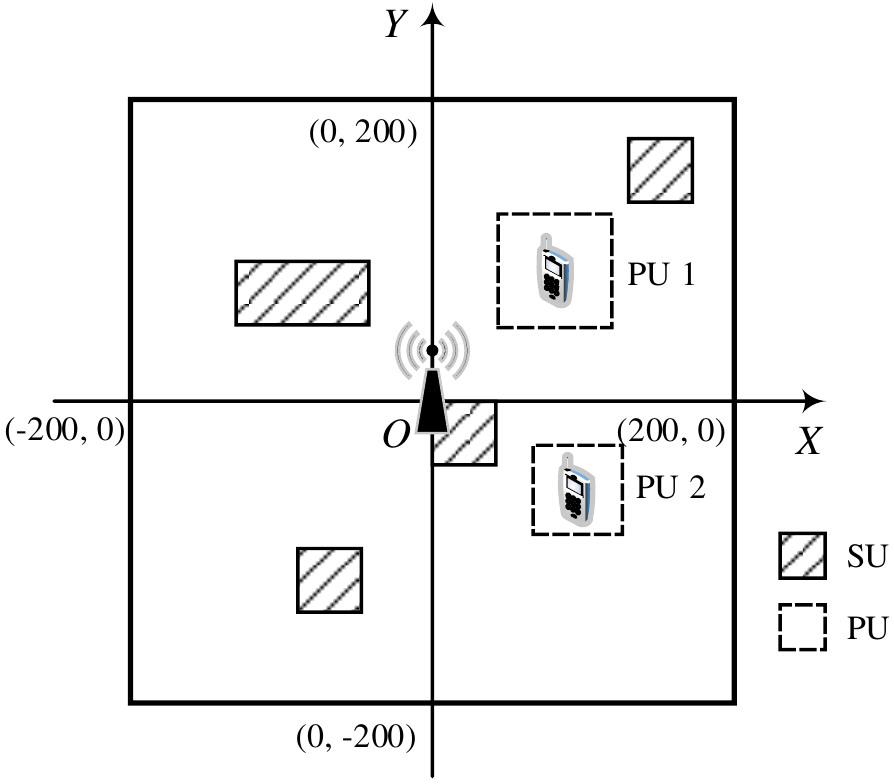}
		\caption{Illustration of the cooperative spectrum sensing scenario.}\label{simScenario}\vspace{-0.5cm}
	\end{figure}

	Firstly, we show the convergence performance of both the VFL algorithm and the centralized learning algorithm in Fig. \ref{NetworkI}. We first use a neural architecture denoted as Network I, where a $3$-layer FCN of size $(203, 2048, 8)$ is used at each SU and a $3$-layer FCN of size $(32, 24, 8)$ is deployed at the server. The step-size is $\eta=10^{-4}$. The activation ratio of SU $1$ is set as $\epsilon_{1,t} = \alpha$ with $\alpha = \{0.9, 0.8, 0.6, 0.4, 0.2, 0.1\}$. It can be observed from Fig. \ref{NetworkI}(a) that when the activation ratio reduces, it requires more number of rounds to achieve a given target accuracy, which coincides with the result in Theorem \ref{ConvergenceTheo}. The centralized learning algorithm, which is marked in blue, gives the best performance with the fastest descent. On the other hand, in Fig. \ref{NetworkI}(b), we show the test accuracy versus the training latency. We can observe that as the activation ratio reduces from $1$ to $0$, the training latency first reduces due to the reduction on the communication latency, then it turns to increase when the activation ratio is close to $0$, because truncating most of the SUs highly increases the number of round required to converge. Note that an extreme activation ratio, e.g., $\epsilon_{1,t} = 0.1$, gives slower convergence compared with the baseline, which indicates that reducing the activation ratio is not always beneficial for the training latency reduction under the neural architecture Network I.

	Following the design rules in Sec. \ref{architectureDesign}, we further propose another neural architecture Network II, where a $3$-layer FCN of size $(203, 32, 8)$ is used at each SU and a $3$-layer FCN of size $(32, 512, 8)$ is deployed at the server. The convergence performance of Network II is shown in Fig. \ref{NetworkII}. Compared with Network I, Network II features a larger central model and smaller local models, which shifts more computation load from the SUs to the server. The architecture of Network II follows the Design Rule 2 in Sec. \ref{architectureDesign}, where a large central model makes $v_{0,t}$ dominant in $v$. Such a design leads to a favorable property that the number of training rounds required to achieve a given testing accuracy is not sensitive to the change in activation ratio. In this case, reducing the activation ratio directly reduces the communication latency without much effect on the number of required training rounds. Therefore, according to \eqref{TotalBound}, by using the proposed architecture Network II, choosing the activation ratio within a large range, say $(0.1, 0.8)$, always gives a better performance compared to the baseline, which is not the case for Network I.  Moreover, in Fig. \ref{NetworkII}, different activation ratios give similar convergence rates, indicating that an arbitrary activation ratio gives a near-optimal solution. Therefore, the proposed neural architecture Network II is more robust than Network I under the practical situation where the optimal activation ratio is unknown a prior.
	
	\begin{figure}[t]
		\centering
		\subfigure[MSE versus the number of rounds for  VFL algorithms using Network I.]{\includegraphics[width=4.3cm]{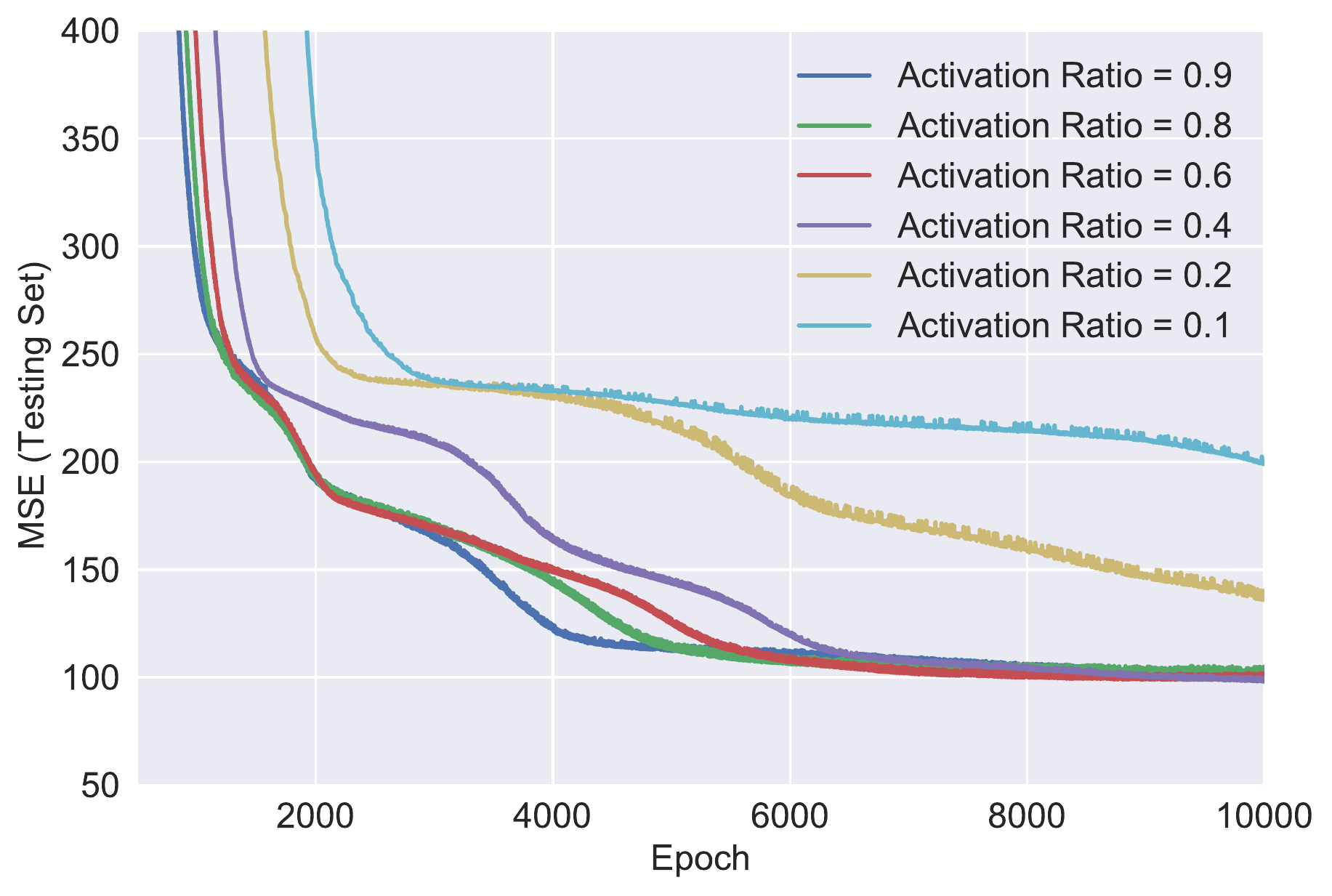}}
		\subfigure[MSE versus training latency for VFL algorithms using Network I.]{\includegraphics[width=4.3cm]{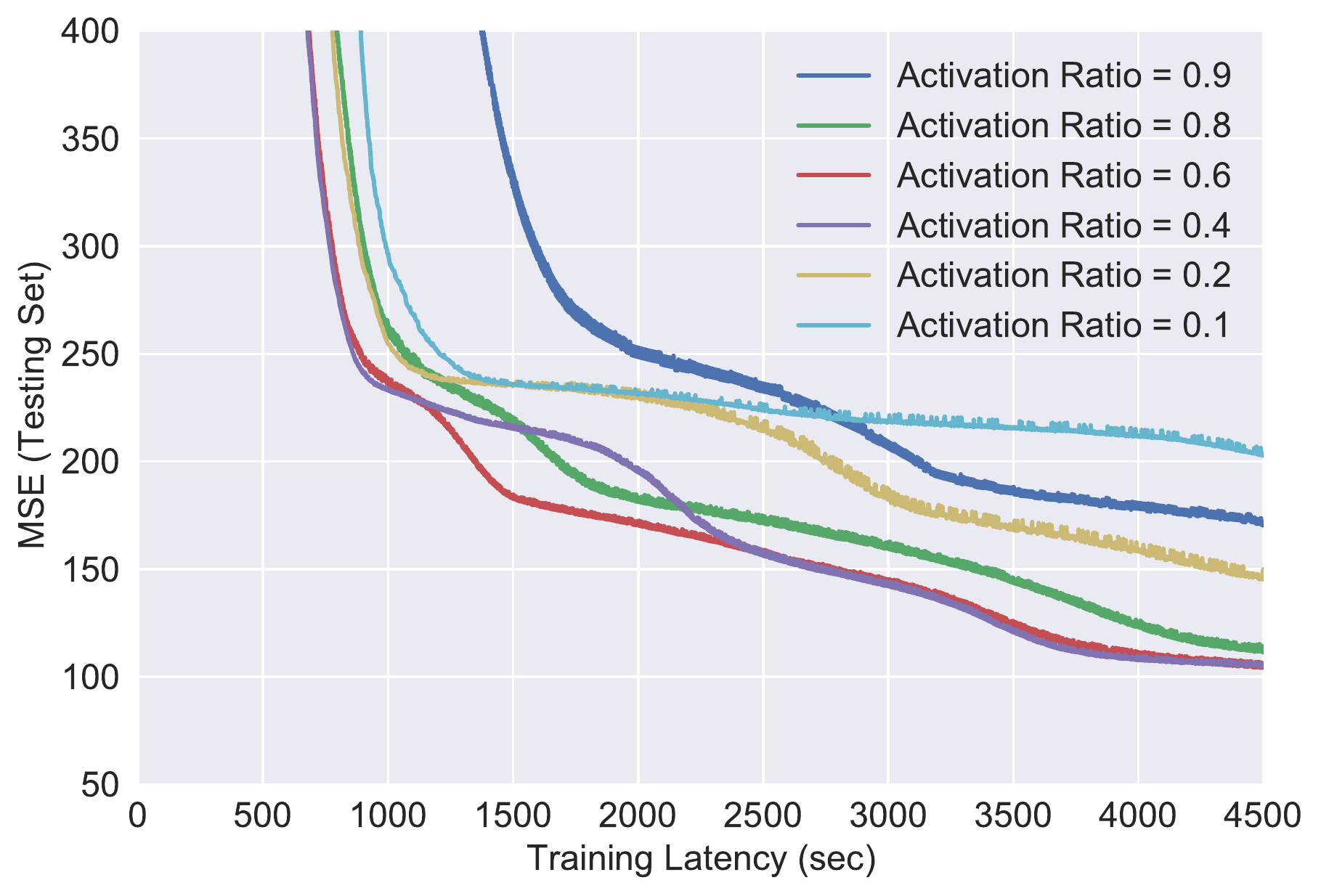}}
		\caption{Convergence performance of VFL algorithms with different activation ratios by using Network I.}\label{NetworkI}\vspace{-0.5cm}
	\end{figure}
	
	\begin{figure}[t]
		\centering
		\subfigure[MSE versus the number of rounds for  VFL algorithms using Network II.]{\includegraphics[width=4.3cm]{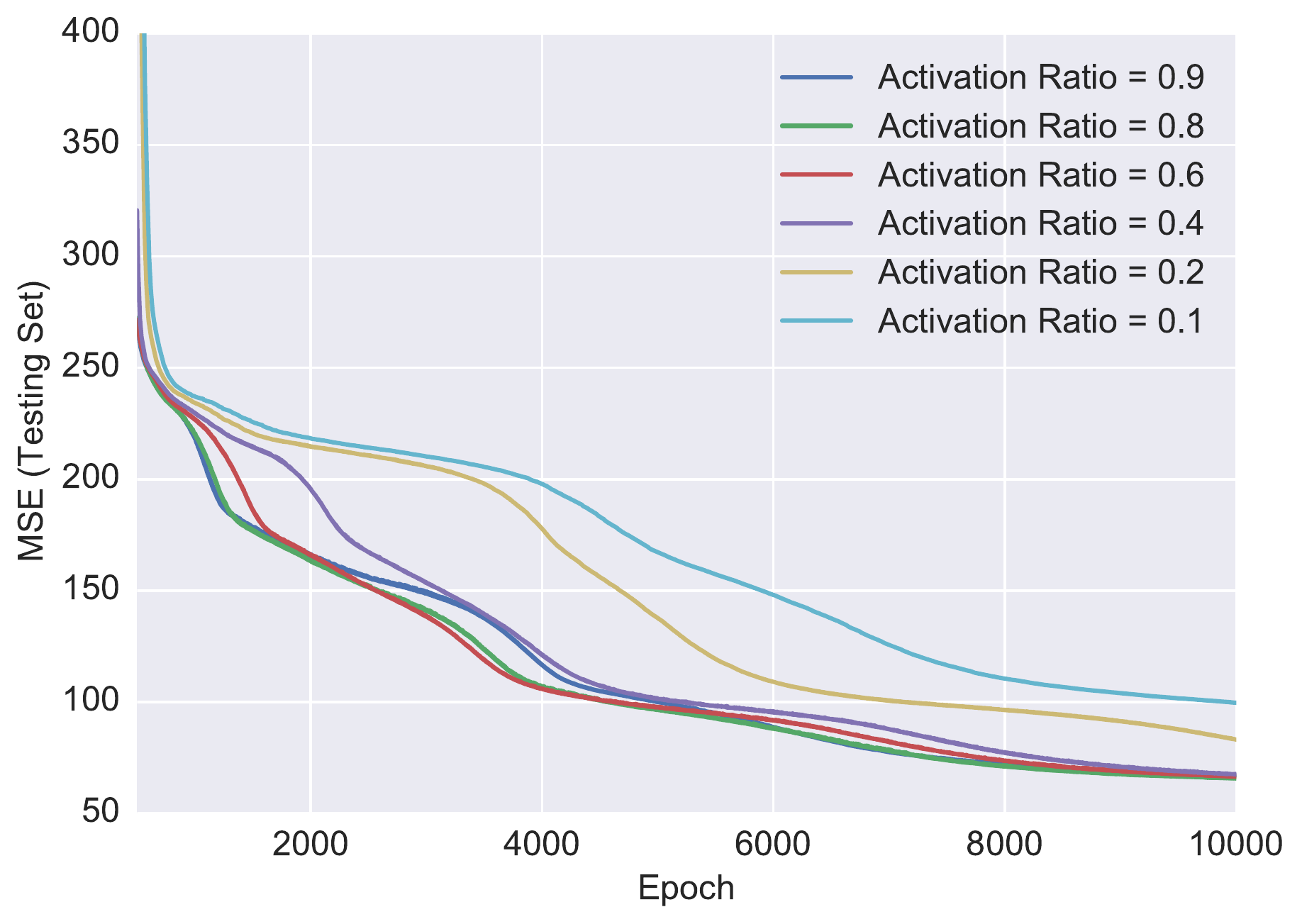}}
		\subfigure[MSE versus training latency for VFL algorithms using Network II.]{\includegraphics[width=4.3cm]{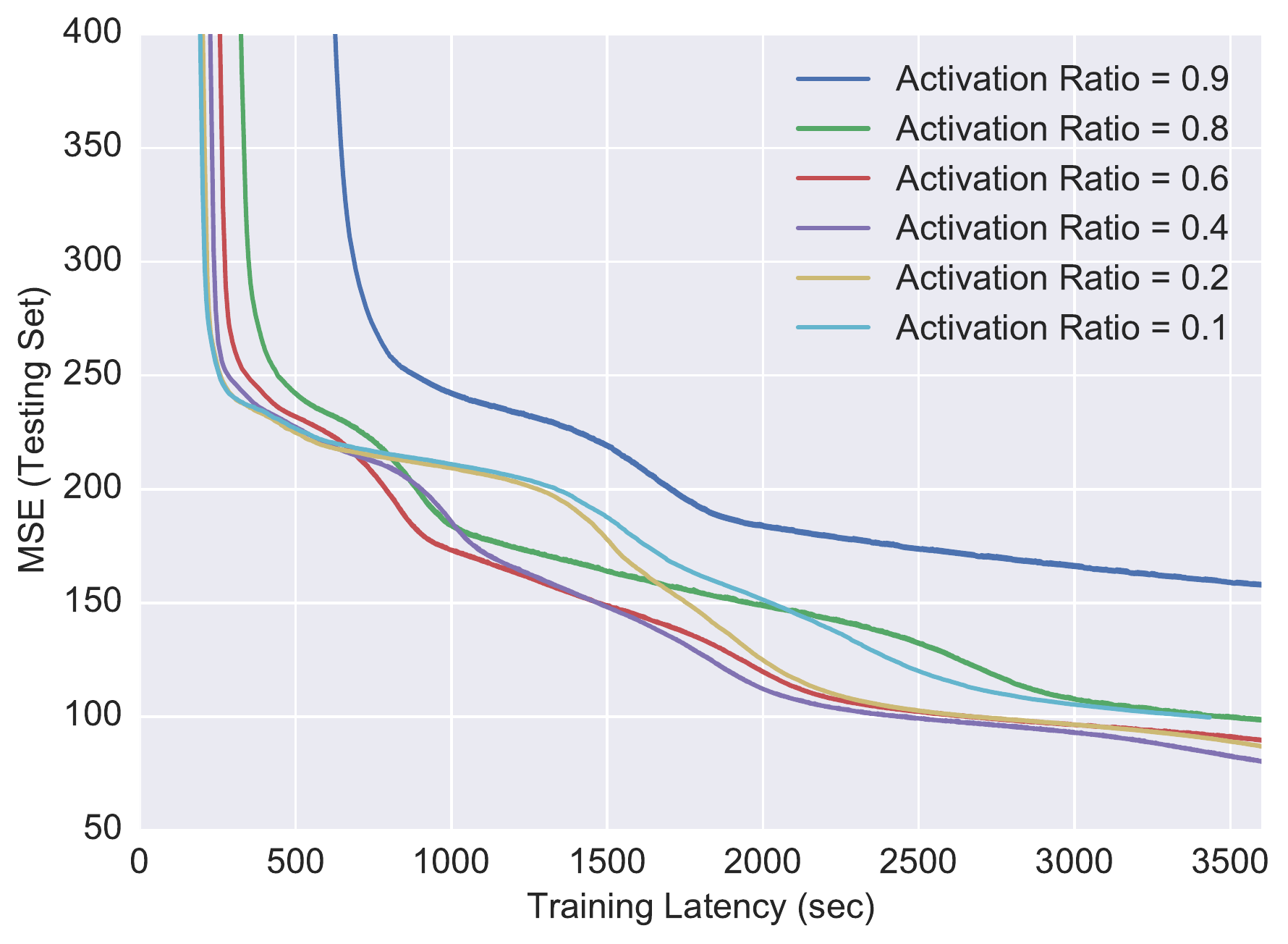}}
		\caption{Convergence performance of VFL algorithms with different activation ratios by using Network II .}\label{NetworkII}\vspace{-0.5cm}
	\end{figure}

    In the next, we show the detection accuracy of the proposed VFL algorithm on both PUs' transmit power levels and their locations. 
    As shown in Fig. \ref{NetworkII}(a), the MSE is around $60$ after $10000$ training rounds by using Network II, which mainly comes from the location detection. To show the detection error clearly, we compare the estimated PUs' state and the ground-truth labels using $20$ data samples from the testing set. The detection error of the transmit power levels and the 3D locations are presented in  Fig. \ref{ResolutionStrength1} and Fig. \ref{LocationDeviation}, respectively. It can be observed that the estimation error on the transmit power levels is quite small. Moreover, the results in Fig. \ref{LocationDeviation} indicate that an increase in the number of SUs can effectively suppress the detection error, which corresponds to the \emph{fusion gain} under the VFL framework as more features are observed.
    
    \begin{figure}[t]
    	\centering
    	\subfigure[Comparison between the estimated transmit power level of PU 1 and the labels.]{\includegraphics[height=3cm, width=4.3cm]{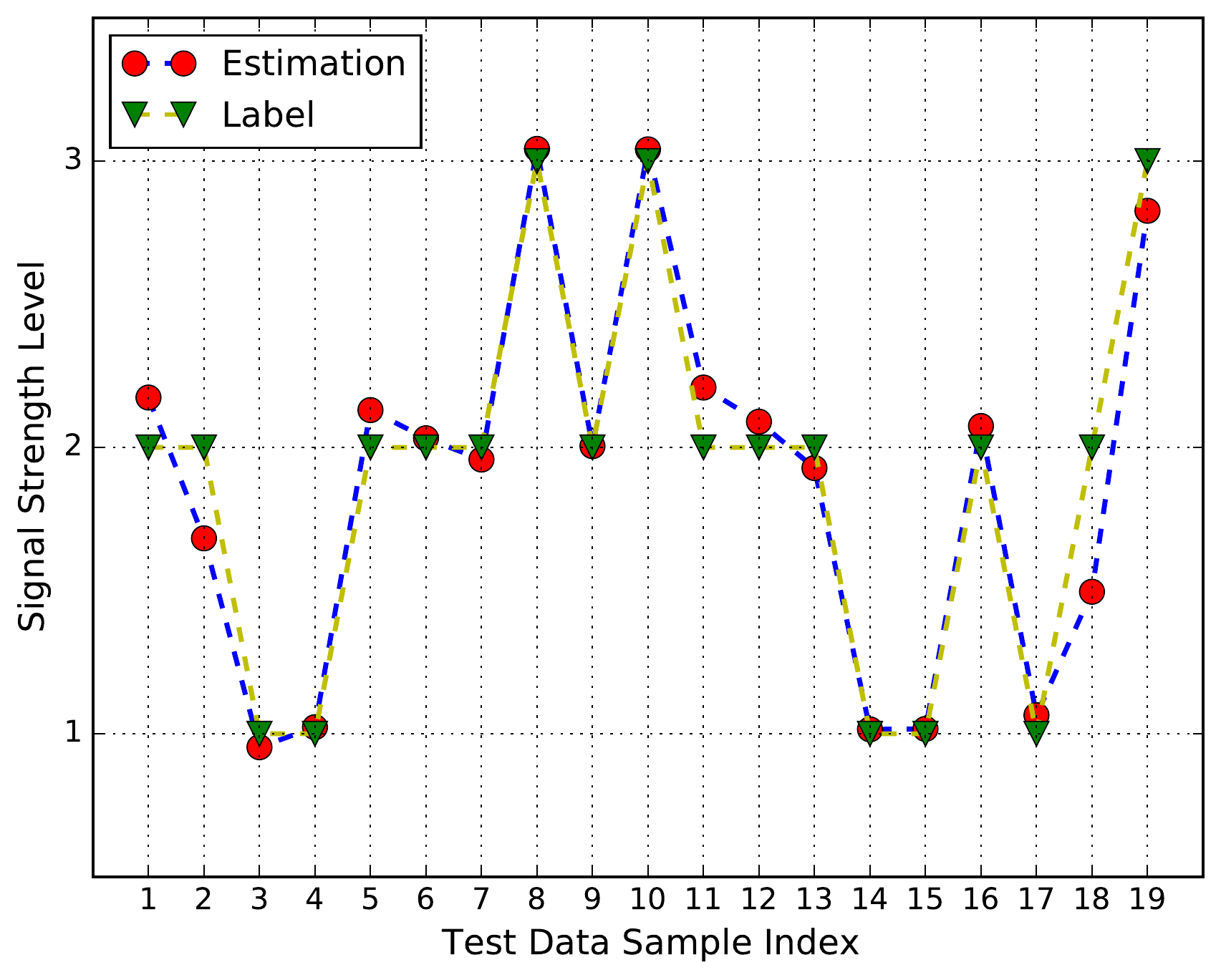}}
    	\subfigure[Comparison between the estimated transmit power level of PU 2 and the labels.]{\includegraphics[height=3cm, width=4.3cm]{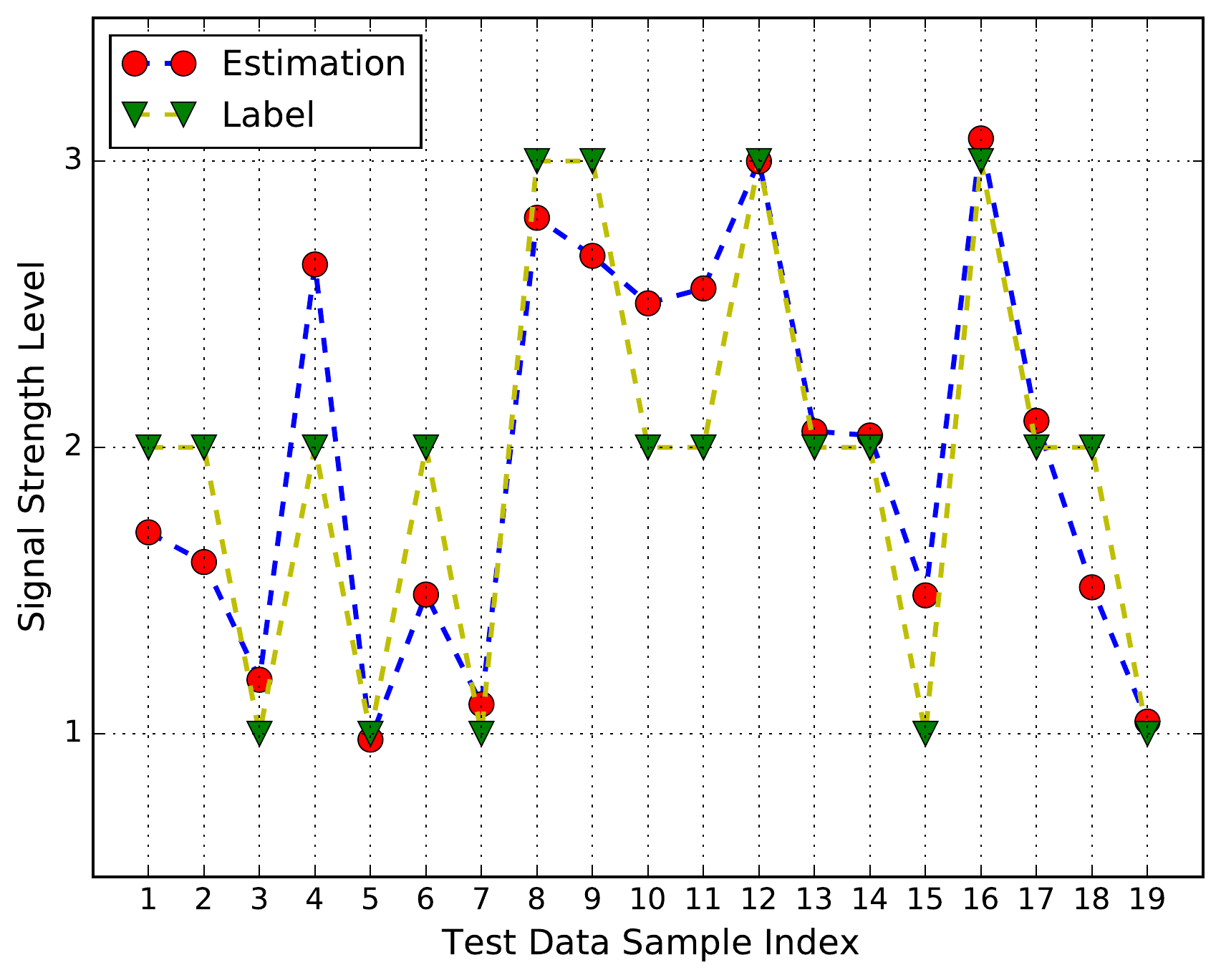}}
    	\caption{Accuracy of the detection on the transmit power level of PU 1.}\label{ResolutionStrength1}\vspace{-0.5cm}
    \end{figure}

    \begin{figure}[t]
    	\centering
    	\subfigure[Location detection error with observations from 4 SUs.]{\includegraphics[height=3cm, width=4.3cm]{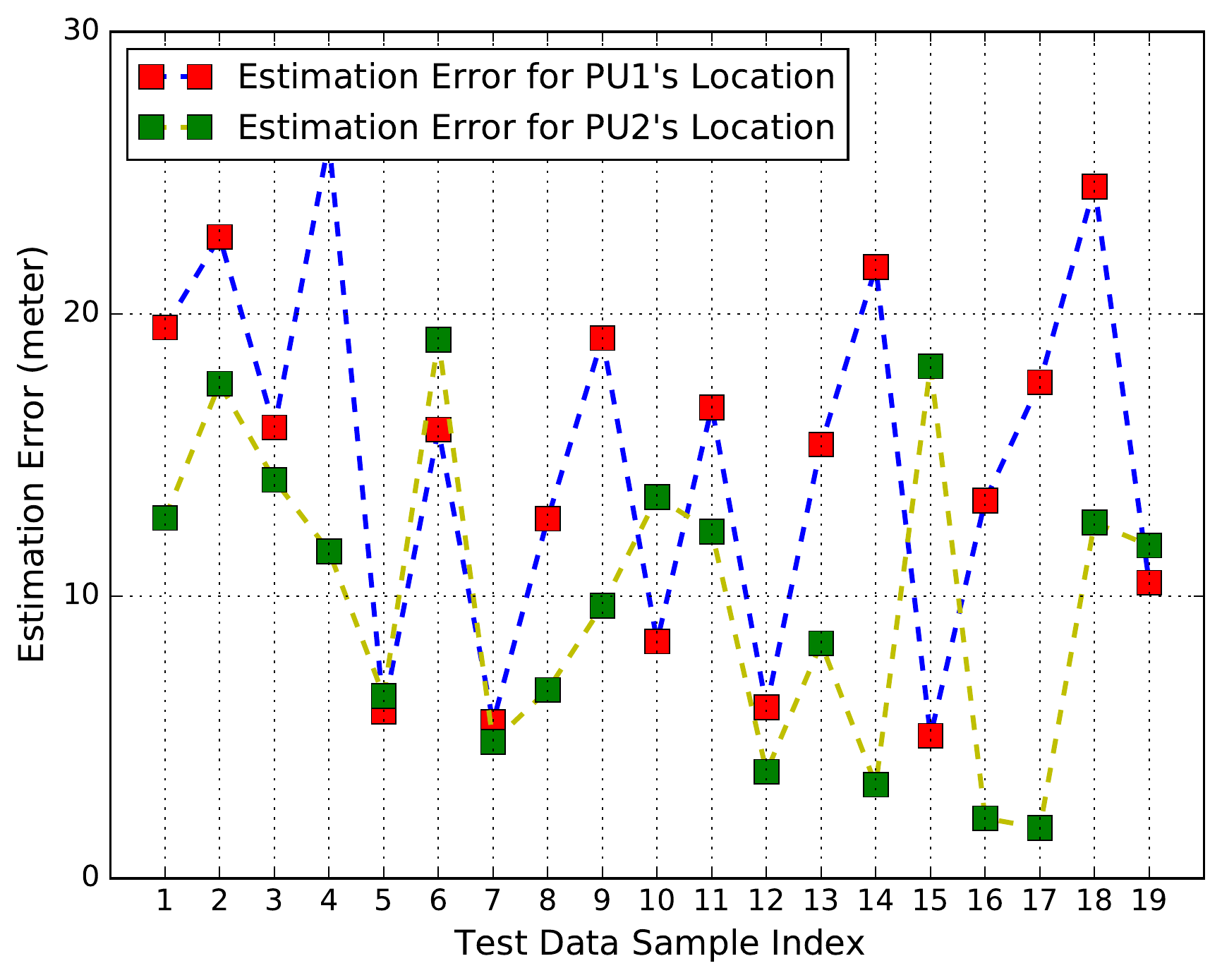}}
    	\subfigure[Location detection error with observations from 8 SUs.]{\includegraphics[height=3cm, width=4.3cm]{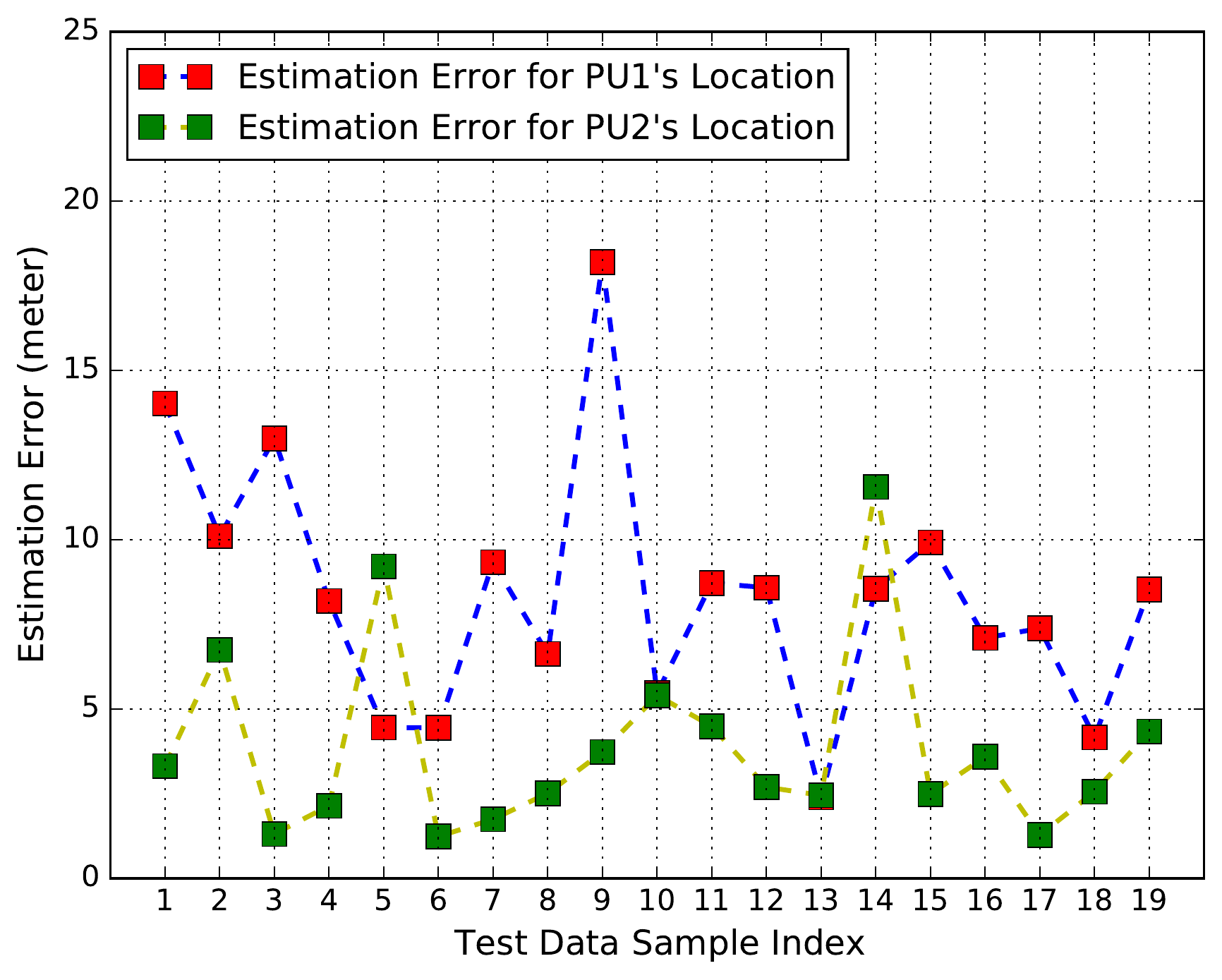}}
    	\caption{3D location detection error of the 2 PUs.}\label{LocationDeviation}\vspace{-0.5cm}
    \end{figure}

	
		\section{Appendix}
	\subsection{Proof of Theorem \ref{ConvergenceTheo}}\label{A}
	Given the $L$-smoothness in \eqref{smoothness}, the loss function follows
	\begin{align}\label{TaylorSmoothness}
	\mathcal{L}(\boldsymbol{\Theta}_{t+1}) \!\le\! \mathcal{L}(\boldsymbol{\Theta}_t) \!+\! (\boldsymbol{\Theta}_{t+1}\!-\!\boldsymbol{\Theta}_t)\nabla\! \mathcal{L}(\boldsymbol{\Theta}_t) \!+\! \frac{L}{2} \|\boldsymbol{\Theta}_{t+1} \!-\! \boldsymbol{\Theta}_t\|^2,
	\end{align}
	according to the Taylor expansion.
	In the same way, the $\mu$-convexity also gives a result as
	\begin{align}\label{TaylorConvexity}
	\mathcal{L}(\boldsymbol{\Theta}_{t+1}) \!\ge\! \mathcal{L}(\boldsymbol{\Theta}_t) \!+\! (\boldsymbol{\Theta}_{t+1}\!-\!\boldsymbol{\Theta}_t)\nabla \mathcal{L}(\boldsymbol{\Theta}_t) \!+\! \frac{\mu}{2} \|\boldsymbol{\Theta}_{t+1} \!-\! \boldsymbol{\Theta}_t\|^2.
	\end{align}
	By using gradient descent to update the model parameters, we have
	\begin{align}\label{step1}
	\mathcal{L}(\boldsymbol{\Theta}_{t+1}) \le& \mathcal{L}(\boldsymbol{\Theta}_t) - \eta (\nabla \mathcal{L}(\boldsymbol{\Theta}_t) - \mathbf{e} - \mathbf{o})\nabla \mathcal{L}(\boldsymbol{\Theta}_t) \nonumber\\
	&+ \frac{L\eta^2}{2} \|\nabla \mathcal{L}(\boldsymbol{\Theta}_t) - \mathbf{e}- \mathbf{o}\|^2,
	\end{align}
	following the result in \eqref{TaylorSmoothness}, where $\mathbf{e}$ is the error in gradient due to the absence of unscheduled devices in round $t$ and $\mathbf{o}$ is the error introduced by the finite training dataset or mini-batch.
	Specifically, when an active set of users $\mathcal{K}_A(t)$ is scheduled in round $t$, according to the above equation \eqref{step1}, the expected optimization function  $\mathbb{E}[\mathcal{L}(\boldsymbol{\Theta}_{t+1})]$ is bounded by
	\begin{align}\label{step2}
	&\mathbb{E}[\mathcal{L}(\boldsymbol{\Theta}_{t+1})] \nonumber\\
	\le& \mathbb{E}\bigg[\mathcal{L}(\boldsymbol{\Theta}_t) \!-\! \eta\!\!\! \sum\limits_{k \in \mathcal{K}_A(t) \cup \{0\}} \!\!\!\big\|\nabla_k \mathcal{L}(\boldsymbol{\Theta}_t)\big\|^2 \nonumber\\
	&+ \frac{L\eta^2}{2}\!\!\! \sum\limits_{k \in \mathcal{K}_A(t) \cup \{0\}}\!\!\! \big\|\nabla_k \mathcal{L}(\boldsymbol{\Theta}_t)\big\|^2\bigg] \!\!+\! \frac{L\eta^2}{2}\mathbb{E}[\|\mathbf{o}\|]^2 \nonumber\\
	=& \mathbb{E}[\mathcal{L}(\boldsymbol{\Theta}_t)] - \eta v_t (1 - \frac{L\eta}{2}) \mathbb{E} [\|\nabla \mathcal{L}(\boldsymbol{\Theta}_t)\|^2] + \frac{L\eta^2 c}{2},
	\end{align}
	where $\nabla_k \mathcal{L}(\boldsymbol{\Theta}_t) = \frac{\partial \mathcal{L}(\boldsymbol{\Theta}_t)}{\partial \boldsymbol{\theta}_k}$ is the gradient of parameters in device $k$'s local model, and $\|\nabla \mathcal{L}(\boldsymbol{\Theta}_t)\|^2 = \sum\limits_{k \in \mathcal{K} \cup \{0\}}\!\! \|\nabla_k \mathcal{L}(\boldsymbol{\Theta}_t)\|^2$ is the gradient power of the whole DNN.
	
	Moreover, by minimizing both sides of \eqref{TaylorConvexity} in terms of $\boldsymbol{\Theta}$, we have the following inequality
	\begin{align}\label{step3}
	&\min\limits_{\boldsymbol{\Theta}_{t+1}}\mathcal{L}(\boldsymbol{\Theta}_{t+1}) \nonumber\\
	\!\!\!\!\!\!\!\!\ge &\min\limits_{\boldsymbol{\Theta}_{t+1}} \mathcal{L}(\boldsymbol{\Theta}_t) \!+\! (\boldsymbol{\Theta}_{t+1}\!-\!\boldsymbol{\Theta}_t)\nabla \!\mathcal{L}(\boldsymbol{\Theta}_t) \!+\! \frac{\mu}{2} \|\boldsymbol{\Theta}_{t+1} \!-\! \boldsymbol{\Theta}_t\|^2.
	\end{align}
	Note that the RHS of \eqref{step3} is a convex function and thereby the minimizer can be straightforward obtained as $-\frac{1}{\mu} \nabla \mathcal{L}(\boldsymbol{\Theta}_t)$. Then equation \eqref{step3} can be simplified as
	\begin{align}
	\mathcal{L}(\boldsymbol{\Theta}^\ast) \ge \mathcal{L}(\boldsymbol{\Theta}_t) - \frac{1}{2\mu} \|\nabla\mathcal{L}(\boldsymbol{\Theta}_t)\|^2, 
	\end{align}
	which further gives
	\begin{align}\label{step4}
	\mathbb{E}[\|\nabla\mathcal{L}(\boldsymbol{\Theta}_t)\|^2] \ge 2\mu \mathbb{E}\left[ \mathcal{L}(\boldsymbol{\Theta}_t) - \mathcal{L}(\boldsymbol{\Theta}^\ast) \right].
	\end{align}
	By plugging \eqref{step4} into \eqref{step2}, we have
	\begin{align}
	&\mathbb{E}[\mathcal{L}(\boldsymbol{\Theta}_{t+1}) \!-\! \mathcal{L}(\boldsymbol{\Theta}^\ast)] \nonumber\\
	\le& \mathbb{E}[\mathcal{L}(\boldsymbol{\Theta}_t) \!-\! \mathcal{L}(\boldsymbol{\Theta}^\ast)] - \eta v_t (1 \!-\! \frac{L\eta}{2}) \mathbb{E}[\|\nabla \mathcal{L}(\boldsymbol{\Theta}_t)\|^2] \!+\! \frac{L\eta^2 c}{2} \nonumber\\
	\le& \left[ 1 \!-\! 2\mu \eta v_t (1 \!-\! \frac{L\eta}{2})\right]  \mathbb{E}\left[ \mathcal{L}(\boldsymbol{\Theta}_t) \!-\! \mathcal{L}(\boldsymbol{\Theta}^\ast) \right] + \frac{L\eta^2 c}{2} \nonumber\\
	\le& \left[ 1 \!-\! 2\mu \eta v (1 \!-\! \frac{L\eta}{2})\right]^{t+1}  \mathbb{E}\left[ \mathcal{L}(\boldsymbol{\Theta}_0) \!-\! \mathcal{L}(\boldsymbol{\Theta}^\ast) \right] + \frac{L\eta c}{2\mu v}, 
	\end{align}
	where $v = \min\{v_t, \ \forall v\}$ is a lower bound on the controllable variable $v_t$.
	By choosing appropriately small step-size $\eta \le \frac{1}{L}$, the convergence process can be further described as
	\begin{align}
	&\mathbb{E}[\mathcal{L}(\boldsymbol{\Theta}_{t+1}) - \mathcal{L}(\boldsymbol{\Theta}^\ast)] \nonumber\\
	\le &\left[ 1 - \frac{\mu}{L} v \right]^{t+1}  \mathbb{E}\left[ \mathcal{L}(\boldsymbol{\Theta}_0) - \mathcal{L}(\boldsymbol{\Theta}^\ast) \right]+ \frac{c}{2\mu v}.
	\end{align}


	\bibliographystyle{ieeetr}
    \bibliography{Ref}

\begin{thebibliography}{1}

\bibitem{Hussain2009}
S.~Hussain and X.~Fernando, ``Spectrum sensing in cognitive radio networks:
  Up-to-date techniques and future challenges,'' in {\em 2009 IEEE Toronto
  International Conference Science and Technology for Humanity (TIC-STH)},
  pp.~736--741, Sep. 2009.

\bibitem{Quan2008}
Z.~Quan, S.~Cui, and A.~H. Sayed, ``Optimal linear cooperation for spectrum
  sensing in cognitive radio networks,'' {\em IEEE Journal of Selected Topics
  in Signal Processing}, vol.~2, pp.~28--40, Feb 2008.

\bibitem{Mcmahan17}
B.~McMahan, E.~Moore, D.~Ramage, S.~Hampson, and B.~A.~y. Arcas,
  ``{Communication-Efficient Learning of Deep Networks from Decentralized
  Data},'' in {\em Proceedings of the 20th International Conference on
  Artificial Intelligence and Statistics} (A.~Singh and J.~Zhu, eds.), vol.~54
  of {\em Proceedings of Machine Learning Research}, pp.~1273--1282, PMLR,
  20--22 Apr 2017.

\bibitem{YQ2019}
Q.~Yang, Y.~Liu, T.~Chen, and Y.~Tong, ``Federated machine learning: Concept
  and applications,'' {\em ACM Trans. Intell. Syst. Technol.}, vol.~10, jan
  2019.

\bibitem{Gao2019}
J.~Gao, X.~Yi, C.~Zhong, X.~Chen, and Z.~Zhang, ``Deep learning for spectrum
  sensing,'' {\em IEEE Wireless Communications Letters}, vol.~8,
  pp.~1727--1730, Dec 2019.

\bibitem{Zhang2020}
Y.~Zhang, Q.~Wu, and M.~Shikh-Bahaei, ``Vertical federated learning based
  privacy-preserving cooperative sensing in cognitive radio networks,'' in {\em
  2020 IEEE Globecom Workshops}, pp.~1--6, Dec 2020.

\bibitem{YQarxiv}
Y.~Liu, Y.~Kang, X.~Zhang, L.~Li, Y.~Cheng, T.~Chen, M.~Hong, and Y.~Qiang, ``A
  communication efficient collaborative learning framework for distributed
  features.'' [Online] \url{http://arxiv.org/pdf/1912.11187 .pdf}.

\bibitem{Broadband}
G.~Zhu, Y.~Wang, and K.~Huang, ``Broadband analog aggregation for low-latency
  federated edge learning,'' {\em IEEE Transactions on Wireless
  Communications}, vol.~19, pp.~491--506, Jan 2020.

\bibitem{Overfitting}
D.~M. Hawkins, ``The problem of overfitting,'' {\em Journal of Chemical
  Information and Computer Sciences}, vol.~44, no.~1, pp.~1--12, 2004.
\newblock PMID: 14741005.

\end{thebibliography}
\end{document}